\documentclass[aps,%
notitlepage,%
floatfix,%
aps,%
prl,%
reprint,%
superscriptaddress,%
longbibliography]{revtex4-1}%
\usepackage[utf8]{inputenc}
\usepackage{mathrsfs}       
\usepackage{mathtools}
\usepackage{amsmath}
\usepackage{amsthm}         
\usepackage{amsfonts}       
\usepackage[titletoc,title]{appendix}
\usepackage{algorithm,algorithmic}
\usepackage{bbm}            
\usepackage{bm}             
\usepackage{times}

\usepackage[dvipsnames]{xcolor}
\definecolor{beamer@blendedblue}{rgb}{0.2,0.2,0.7}

\usepackage[colorlinks=true,%
    bookmarksnumbered=true, 
    linkcolor=beamer@blendedblue,%
    bookmarks=true,%
    breaklinks=true,%
    unicode=false, 
    pdfstartview={FitH}, 
    filecolor=beamer@blendedblue,%
    anchorcolor=yellow,%
    citecolor=beamer@blendedblue,%
    urlcolor=beamer@blendedblue,%
    pdfauthor={Kun Wang, Yu-Ao Chen, and Xin Wang},%
    pdfsubject={Mitigating Quantum Errors via Truncated Neumann Series},%
CJKbookmarks=true]{hyperref}

\newtheorem{definition}{Definition}
\newtheorem{proposition}[definition]{Proposition}
\newtheorem{lemma}[definition]{Lemma}
\newtheorem{theorem}[definition]{Theorem}

\mathchardef\ordinarycolon\mathcode`\:
\mathcode`\:=\string"8000
\def\vcentcolon{\mathrel{\mathop\ordinarycolon}}
\begingroup \catcode`\:=\active
  \lowercase{\endgroup
  \let :\vcentcolon
  }

\DeclareFontFamily{U}{mathx}{\hyphenchar\font45}
\DeclareFontShape{U}{mathx}{m}{n}{<-> mathx10}{}
\DeclareSymbolFont{mathx}{U}{mathx}{m}{n}
\DeclareMathAccent{\widebar}{0}{mathx}{"73}

\newcommand{\ket}[1]{\left\vert{#1}\right\rangle}
\newcommand{\bra}[1]{\left\langle{#1}\right\vert}
\newcommand{\ketbra}[2]{\vert{#1}\rangle\!\langle{#2}\vert}
\newcommand\proj[1]{\vert{#1}\rangle\!\langle{#1}\vert}
\newcommand{\opn}[1]{\operatorname{#1}}
\DeclareMathOperator{\tr}{Tr}  
\newcommand{\1}{\mathbbm{1}}
\newcommand{\ox}{\otimes}
\newcommand{\id}{\operatorname{id}}

\newcommand{\norm}[2]{\ensuremath{\left\lVert#1\right\rVert_{#2}}}%
\newcommand{\spectral}[1]{\ensuremath{\left\lVert#1\right\rVert}}%
\makeatletter
\newsavebox{\@brx}
\newcommand{\llangle}[1][]{\savebox{\@brx}{\(\m@th{#1\langle}\)}%
  \mathopen{\copy\@brx\kern-0.5\wd\@brx\usebox{\@brx}}}
\newcommand{\rrangle}[1][]{\savebox{\@brx}{\(\m@th{#1\rangle}\)}%
  \mathclose{\copy\@brx\kern-0.5\wd\@brx\usebox{\@brx}}}
\makeatother

\newcommand*{\cA}{\mathcal{A}}

\newcommand*{\cD}{\mathcal{D}}

\newcommand*{\cM}{\mathcal{M}}
\newcommand*{\cN}{\mathcal{N}}

\newcommand{\bE}{\mathbb{E}}
\newcommand{\bR}{\mathbb{R}}

\begin{document}
\title{\Large\textbf{Mitigating Quantum Errors via Truncated Neumann Series}}

\author{Kun Wang}
\email{wangkun28@baidu.com}
\affiliation{Institute for Quantum Computing, Baidu Research, Beijing 100193, China}

\author{Yu-Ao Chen}
\email{chenyuao@baidu.com}
\affiliation{Institute for Quantum Computing, Baidu Research, Beijing 100193, China}

\author{Xin Wang}
\email{wangxin73@baidu.com}
\affiliation{Institute for Quantum Computing, Baidu Research, Beijing 100193, China}

\begin{abstract}
Quantum gates and measurements on quantum hardware are inevitably subject to hardware
imperfections that lead to quantum errors.
Mitigating such unavoidable errors is crucial to explore the power of quantum hardware better.
In this paper, we propose a unified framework that can mitigate quantum gate and measurement errors in
computing quantum expectation values utilizing the truncated Neumann series.
The essential idea is to cancel the effect of quantum error by approximating its inverse
via linearly combining quantum errors of different orders
produced by sequential applications of the quantum devices with carefully chosen coefficients.
Remarkably, the estimation error decays exponentially in the truncated order,
and the incurred error mitigation overhead is independent of the system size,
as long as the noise resistance of the quantum device is moderate.
We numerically test this framework for different quantum errors and
find that the computation accuracy is substantially improved.
Our framework possesses several vital advantages:
it mitigates quantum gate and measurement errors in a unified manner,
it neither assumes any error structure nor requires the
tomography procedure to completely characterize the quantum
errors, and most importantly, it is scalable.
These advantages empower our quantum error mitigation framework to be efficient and practical
and extend the ability of near-term quantum devices to deliver quantum applications.
\end{abstract}

\maketitle

\textit{\textbf{Introduction.}}---Quantum computers hold great promise for a variety of
scientific and industrial applications~\cite{mcardle2020quantum,cerezo2020variational,bharti2021noisy}.
Nonetheless, the challenges we face are still formidable.
In the current noisy intermediate-scale quantum (NISQ) era~\cite{preskill2018quantum},
quantum computers introduce significant quantum errors that must be dealt with before performing any exhilarating tasks.
Such errors occur either due to unwanted interactions of qubits with the environment
or the physical imperfections of qubit initializations, quantum gates,
and measurements~\cite{kandala2017hardware,arute2019quantum,arute2020hartree,ai2021exponential}.
The problem can be theoretically resolved with
quantum error correction~\cite{shor1995scheme,steane1996error,calderbank1996good,aharonov2008fault,nielsen2010quantum},
which is far beyond the reach of NISQ quantum computers.
This motivates the question of alleviating quantum errors and increasing
the quantum computation accuracy without quantum error correction.

Quantum error mitigation~\cite{temme2017error} provides an inspirational solution to this question
and has been experimentally implemented~\cite{kandala2019error,song2019quantum,zhang2020error,kim2021scalable}.
Typically, errors in a quantum device are classified into quantum gate and measurement errors.
For gate errors, numerous techniques have been designed such as
zero-noise extrapolation~\cite{temme2017error,endo2018practical,dumitrescu2018cloud,he2020resource,giurgica2020digital},
probabilistic error cancellation~\cite{temme2017error,endo2018practical,li2017efficient,takagi2020optimal,jiang2020physical},
subspace expansion~\cite{mcclean2017hybrid,McClean2020Decoding},
virtual distillation~\cite{koczor2021exponential,koczor2021dominant,xiong2021quantum,cai2021resource,huo2021dual,huggins2020virtual},
learning-based method~\cite{czarnik2020error,lowe2020unified,strikis2020learning},
and many others~\cite{mcardle2019error,bonet2018low,endo2021hybrid,sun2021mitigating,
OBrien2020error,cai2021practical,takagi2021fundamental}.
For measurement errors, the focus is not as much as that on gate errors
though measurement errors are significantly larger than gate errors on many quantum platforms.
A well-known strategy is to regard the measurement error as a classical noise model
and handle it via classical post-processing~\cite{chow2012universal,chen2019detector,geller2020rigorous,maciejewski2020mitigation,tannu2019mitigating,nachman2019unfolding,hicks2021readout,bravyi2020mitigating,geller2020efficient,murali2020software,kwon2020hybrid,funcke2020measurement,zheng2020bayesian,maciejewski2021modeling,barron2020measurement,berg2020model,geller2021conditionally,wang2021measurement}.
However, most of the error mitigation techniques require complete characterization of the quantum device
and is not scalable in general. On the other hand, the techniques targeting one type of error
are not directly applicable to the other at large.

We overcome these challenges by proposing a general error mitigation framework
that can reduce both quantum gate and measurement errors in computing the expectation values of quantum observables.
This framework does not require complete characterization of the quantum devices and is theoretically scalable.
The essential idea is to effectively approximate the inverse of quantum error using truncated Neumann series.
Notably, the estimation accuracy of the expectation value is improved in the presence of quantum errors.

\textit{\textbf{Computing the expectation value.}}---A common quantum computation task
is to estimate the expectation value $\tr[O\rho]$ within a specified precision $\varepsilon$
for a given quantum observable $O$ and an $n$-qubit quantum state $\rho:=U\proj{\bm{0}}U^\dagger$
generated by a quantum gate $U$ with initial $n$-qubit state $\ket{\bm{0}}$.
Without loss of generality, we may assume that $O$ is diagonal in the computational basis
and $\norm{O}{2}\leq 1$ where $\norm{\cdot}{2}$ is the matrix $2$-norm.
This task is the building component of multifarious quantum algorithms,
notable practical examples are variational quantum eigensolvers~\cite{peruzzo2014variational,mcclean2016theory},
quantum approximate optimization algorithm~\cite{farhi2014quantum}, and
quantum machine learning~\cite{biamonte2017quantum,havlivcek2019supervised}.

Ideally, $\tr[O\rho]$ can be estimated in the following way. Consider $M$ independent experiments where in each round
we prepare the state $\rho$ using $U$ and measure each qubit in the computational basis as shown in Fig.~\ref{fig:expectation}.
Let $\bm{s}^m\in\{0,1\}^n$ be the outcome in the $m$-th round. Define the empirical mean value
\begin{align}\label{eq:ideal-average}
\eta^{(0)} := \frac{1}{M}\sum_{m=1}^M O(\bm{s}^m)
\end{align}
where $O(\bm{x})$ is the $\bm{x}$-th diagonal element of $O$.
Let $\opn{vec}(\rho)$ be the $2^n$-dimensional column diagonal vector of $\rho$.
Then~\cite{bravyi2020mitigating}
\begin{align}\label{eq:ideal-expectation}
  E^{(0)} := \bE[\eta^{(0)}]
= \sum_{\bm{x}\in\{0,1\}^n}O(\bm{x})\langle \bm{x}\vert\opn{vec}(\rho) = \tr[O\rho],
\end{align}
where $\bE[X]$ is the expectation of the random variable $X$.
Eq.~\eqref{eq:ideal-expectation} implies that $\eta^{(0)}$ is an unbiased estimator of $\tr[O\rho]$.
Furthermore, the standard deviation $\sigma(\eta^{(0)})\leq1/\sqrt{M}$.
By Hoeffding's inequality~\cite{hoeffding1963probability}, $M=2\log(2/\delta)/\varepsilon^2$
would guarantee that $\eta^{(0)}$ approximates $\tr[O\rho]$ within $\varepsilon$ at probability greater than $1 - \delta$,
where $\delta$ is the confidence and all logarithms are in base $2$ throughout this paper.

\begin{figure}
  \centering
  \includegraphics[width=0.45\textwidth]{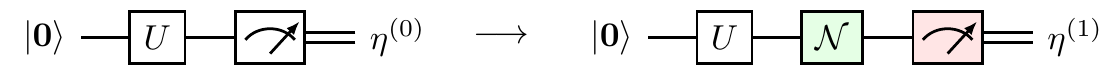}
  \caption{Estimating the expectation value $\tr[O\rho]$ with the ideal quantum devices (left)
        and the noisy quantum devices (right).}
  \label{fig:expectation}
\end{figure}

However, both the quantum gate and measurement of the quantum states suffer from Markovian errors inherent in quantum devices.
For simplicity, we ignore errors in preparing the initial state,
which can be accounted for by regulating noisy quantum gates after qubits initialization.
For the quantum gate error, we assume that the overall evolution is modeled as
the ideal gate evolution $U$ followed by some quantum noisy channel $\cN$~\cite{temme2017error},
i.e., the actual prepared state is $\cN(\rho)$ rather than $\rho$.
The gate error leads to the noisy expectation value $\tr[O\cN(\rho)]$.
Experimentally, the noise channel $\cN$ can be fully characterized via
quantum process or gateset tomography~\cite{greenbaum2015introduction}.
For the measurement error, it was established that such error can be well understood using classical noise
models~\cite{chow2012universal,chen2019detector,geller2020rigorous}.
Specifically, a $n$-qubit noisy measurement device can be characterized
by an error matrix $A$ of size $2^n\times 2^n$.
The element in the $\bm{x}$-th row and $\bm{y}$-th column, $A_{\bm{x}\bm{y}}$,
is the probability of obtaining a outcome $\bm{x}$ provided that the true outcome is $\bm{y}$.
Experimentally, the error matrix can be learnt via calibration~\cite{chen2019detector}.

Suppose now that we adopt the same procedure for computing $\eta^{(0)}$
and obtain the noisy estimator $\eta^{(1)}$ (cf. the right side of Fig.~\ref{fig:expectation}),
where the superscript $1$ indicates that the noisy devices are functioned.
We prove in Appendix~\ref{appx:noisy-expectation} that
\begin{align}\label{eq:noisy-expectation}
E^{(1)} := \bE[\eta^{(1)}] = \sum_{\bm{x}\in\{0,1\}^n}O(\bm{x})\langle \bm{x}\vert A \opn{vec}(\cN(\rho)),
\end{align}
indicating that $\eta^{(1)}$ is no longer an estimator of $\tr[O\rho]$.
Comparing Eqs.~\eqref{eq:ideal-expectation} and~\eqref{eq:noisy-expectation},
we find that in the ideal case, the sampled distribution approximates $\opn{vec}(\rho)$
thanks to the weak law of large numbers,
while in the noisy case, the sampled distribution approximates $A\opn{vec}(\cN(\rho))$
due to both gate and measurement errors, leading to a bias in the estimator.

Our main result is a unified error mitigation framework that can
alleviate the quantum errors identified by $\cN$ and $A$ leading to the noisy estimator $A\opn{vec}(\cN(\rho))$.
In the following, we first describe the general framework and then
specialize it to dealing with quantum gate and measurement errors.

\textit{\textbf{Truncated Neumann series.}}---Let $\bR$ be the real field, $\mathbf{M}_d$ be the
set of $d\times d$ real square matrices and $I\in\mathbf{M}_d$ be the identity matrix.
Let $f:\mathbf{M}_d\to\bR$ be a linear function and $A\in\mathbf{M}_d$ be a matrix.
We are interested in estimating $f(I)$ given access to $f(A)$.
This abstract task encapsulates many important computational tasks in both classical and quantum computing
including the expectation value estimation task described above.
We can show that, under certain conditions, the target $f(I)$ can be efficiently
approximated via linear combination of accessible terms $f(A^k)$ of different orders
with carefully chosen coefficients, employing the Neumann series expansion~\cite[Theorem 4.20]{stewart1998matrix}.
Since $I=AA^{-1}$, the essential idea is to cancel the effect of the inverse $A^{-1}$ via a suitable truncated Neumann series.
The proof is given in Appendix~\ref{appx:Neumann-expansion}.
\begin{proposition}\label{prop:Neumann-expansion}
If $\spectral{I-A}<1$ in some consistent norm, then
\begin{align}\label{eq:Neumann-expansion}
  \left\vert f(I) - \sum_{k=1}^{K+1}c_K(k-1) f\left(A^k\right) \right\vert
= \left\vert f\left((I-A)^{K+1}\right) \right\vert,
\end{align}
where the coefficient function is defined as
\begin{align}\label{eq:c_k}
    c_K(k) := (-1)^{k}\binom{K+1}{k+1},
\end{align}
and $\binom{n}{k}$ is the binomial coefficient.
\end{proposition}

Intuitively, Eq.~\eqref{eq:Neumann-expansion} indicates that
one may approximate the target value $f(I)$ using the first $K$
truncated Neumann terms $f(A^k)$, if
1) the precondition $\spectral{I-A}<1$ is satisfied,
2) the $k$-th order value $f(A^k)$ can be obtained in a similar way as that of $f(A)$, and
3) the remaining term $f((I-A)^{K+1})$ can be upper bounded theoretically.
And, even better, if the upper bound decays exponentially with $K$, the approximation quickly converges.
We show that with appropriate representations all these three conditions can be
satisfied in the quantum error mitigation tasks.
This idea has previously been applied for linear data detection
in massive multiuser multiple-input multiple-output wireless systems~\cite{wu2013approximate}.

\textit{\textbf{Gate Error Mitigation (GEM).}}---We illustrate how to use the Neumann series
framework to mitigate gate errors.
For this aim, we first recall the Pauli transfer matrix (PTM) representation in Appendix~\ref{appx:pauli_transfer_matrix_representation}.
In PTM, quantum states $\vert\rho\rrangle$ and observables $\llangle O\vert$ are represented by vectors
and quantum channels $[\cN]$ are represented by real matrices. For $n$ qubits, vectors and matrices are $4^n$-dimensional.
The expected value of the observable $O$ in the state $\rho$ going through the noisy quantum channel $\cN$ reads as follows:
\begin{align}
    \tr[O\cN(\rho)] = \llangle O \vert [\cN] \vert \rho \rrangle.
\end{align}

Setting $f\equiv\tr$ and $A\equiv[\cN]$, the above task fits into the truncated Neumann series framework.
Define the noise resistance of the quantum channel $\cN$ as
\begin{align}\label{eq:noise resistance gem}
   \xi_g(\cN) := \norm{I - [\cN]}{\infty},
\end{align}
where $\norm{\cdot}{\infty}$ is the matrix $\infty$-norm.
We assume that $\xi_g(\cN)<1$, which is a sufficient condition so that Proposition~\ref{prop:Neumann-expansion} holds.
For Pauli noise $\cN$, $\xi_g<1$ corresponds to the case that the
all Pauli eigenvalues must be strictly positive~\cite{harper2020efficient,chen2021quantum},
which can be satisfied whenever the noise is weak~\cite{flammia2020efficient}.
We show in the following theorem that the approximation error (RHS. of Eq.~\eqref{eq:Neumann-expansion})
can be exponentially upper bounded in terms of $\xi_g$.
The proof is given in Appendix~\ref{appx:GEM}.

\begin{theorem}\label{thm:GEM}
Assume that $\xi_g(\cN) < 1$. For arbitrary positive integer $K$, it holds that
\begin{align}\label{eq:GEM}
    \left\vert \tr[O\rho] - \sum_{k=1}^{K+1}c_K(k-1) E_g^{(k)}\right\vert
\leq \norm{\llangle O\vert}{\infty}\xi_g^{K+1},
\end{align}
where $E_g^{(k)} := \llangle O\vert[\cN]^k\vert\rho\rrangle$.
\end{theorem}

As evident from Theorem~\ref{thm:GEM}, the noise resistance $\xi_g$ of $\cN$
uniquely determines the number of terms required in the truncated Neumann series
to approximate $\tr[O\rho]$ to the desired precision.
What is more, since $\xi_g<1$, the approximation error decays exponentially in terms of $K$,
indicating that small $K$ suffices to reach high estimating accuracy.
Thanks to the multiplicativity property of PTM, which states that the PTM of $\cN^{\circ k}$ is exactly $[\cN]^{k}$,
each $E_g^{(k)}$ can be viewed as the noisy expectation value generated by the noisy gate device executed $k$ times sequentially.
Since measurement errors can be handled independently, we do not concern such errors in GEM.
Let $\overline{E}_g:=\sum_{k=1}^{K+1}c_K(k-1) E_g^{(k)}$.
Theorem~\ref{thm:GEM} inspires a systematic way to estimate the expectation value $\tr[O\rho]$.
Firstly, we choose $K$ so that the RHS. of~\eqref{eq:GEM} evaluates to the desired precision $\varepsilon$, yielding
the optimal gate truncated number
\begin{align}\label{eq:K-opt-gate}
    K_g = \left\lceil\frac{\log\varepsilon - \log\norm{\llangle O\vert}{\infty}}{\log\xi_g} - 1\right\rceil.
\end{align}
Secondly, we compute $\overline{E}_g$ by estimating each $E_g^{(k)}$ and linearly combining them with coefficients $c_K$.
Since $\overline{E}_g$ itself is only an $\varepsilon$-estimate
of $\tr[O\rho]$, it suffices to
approximate $\overline{E}$ within an error $\varepsilon$.
Motivated by the relation between $\eta^{(1)}$ and $E^{(1)}$ in~\eqref{eq:noisy-expectation},
we propose the following procedure to estimate $E_g^{(k)}$ for arbitrary $1\leq k\leq K+1$:
\begin{enumerate}
  \item Generate a quantum state $\rho$.
  \item Execute the channel $\cN$ \emph{sequentially} $k$ times, yielding the final
        state $\cN^{\circ k}(\rho)$. Measure the final state and collect the measurement outcome.
  \item Repeat the above two steps $M$ rounds.
  \item Output the average $\eta_g^{(k)}$ as an estimate of $E_g^{(k)}$.
\end{enumerate}
We claim that the average $\overline{\eta}_g:=\sum_{k=1}^{K+1}c_K(k-1) \eta_g^{(k)}$
approximates $\tr[O\rho]$ within error $2\varepsilon$ with high probability, i.e.,
\begin{align}\label{eq:good-estimate}
    \opn{Pr}\left\{\vert\tr[O\rho] - \overline{\eta}_g\vert\leq 2\varepsilon\right\} \geq 1 - \delta.
\end{align}
The proof is given in Appendix~\ref{appx:good-estimate}. For illustrative purpose,
we demonstrate in Fig.~\ref{fig:expectation-k3} the experimental setup for estimating
the noisy expectation value $E_g^{(3)}$,
where the noisy device is repeated sequentially three times in each round.

\begin{figure}[!htbp]
  \centering
  \includegraphics[width=0.4\textwidth]{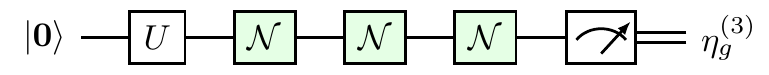}
  \caption{Experimental setup for estimating $E_g^{(3)}$,
        in which the noisy gate device (box in blue) is executed $3$ times sequentially.}
  \label{fig:expectation-k3}
\end{figure}

We emphasize that the GEM method does not rely on the complete characterization of $\cN$,
at least in theory, since the only relevant quantity is the noise resistance $\xi_g$,
which we believe can be estimated efficiently without tomography.
Furthermore, the repetition of $\cN$ can be experimentally accomplished by inserting
identity gates after the noisy gate under the assumption that all gates incur the same noise.

\textit{\textbf{Measurement Error Mitigation (MEM).}}---Recall that the error of a measurement device
is well understood using classical noise models and is characterized by an error matrix $A$.
It is straightforward to classically reverse the noise effects
by multiplying the sampled distribution by the inversion $A^{-1}$.
Nonetheless, there are several limitations of the direct inverse approach:
(i) Completely characterizing $A$ requires $2^n$ calibration setups and is not scalable;
(ii) $A$ may be singular which prevents direct inversion; and
(iii) $A^{-1}$ is hard to compute in general and might not be column stochastic, resulting unphysical estimates.

We manifest how to use the Neumann series framework to mitigate measurement errors while avoiding the limitations.
Similar to the GEM case, the essential idea of MEM is to effectively simulate the inverse of the error matrix $A$,
utilizing the truncated Neumann series.
First of all, we define the noise resistance of the error matrix $A$ as
\begin{align}\label{eq:noise resistance mem}
   \xi_m(A) := 2\left(1 - \min_{\bm{x}\in\{0,1\}^n}\bra{\bm{x}}A\ket{\bm{x}}\right).
\end{align}
By definition, $1-\xi_m/2$ is the minimal diagonal element of $A$.
Intuitively, $\xi_m/2$ characterizes the measurement device's worst-case behavior
since it is the maximal probability for which the true and actual outcomes mismatch.
In the following, we assume $\xi_m<1$, which is equivalent to the condition that
the minimal diagonal element of $A$ is larger than $0.5$.
This assumption is reasonable since otherwise
the measurement device is too noisy to be applicable from the practical perspective.
It is also a sufficient condition under which Proposition~\ref{prop:Neumann-expansion} holds.
We deploy the truncated Neumann series framework to alleviate the measurement errors and obtain the following,
in a similar favor of Theorem~\ref{thm:GEM}. The proof is given in Appendix~\ref{appx:MEM}.

\begin{theorem}\label{thm:MEM}
Assume that $\xi_m<1$. For arbitrary positive integer $K$, it holds that
\begin{align}\label{eq:MEM}
    \left\vert \tr[O\rho] - \sum_{k=1}^{K+1}c_K(k-1) E_m^{(k)}\right\vert \leq \xi_m^{K+1},
\end{align}
where
\begin{align}\label{eq:E_k}
    E_m^{(k)} := \sum_{\bm{x}\in\{0,1\}^n}O(\bm{x})\langle \bm{x}\vert A^{k}\opn{vec}(\rho).
\end{align}
\end{theorem}

Upper bounding the RHS. of~\eqref{eq:MEM} with the desired precision $\varepsilon$
yields the optimal measurement truncated number
\begin{align}\label{eq:K-opt-mem}
    K_m = \left\lceil\frac{\log\varepsilon}{\log\xi_m} - 1\right\rceil.
\end{align}
We can evaluate the noisy values $E_m^{(k)}$ up to the optimal order $K_m$ in almost the same manner
as we estimated $E_g^{(k)}$ in GEM. The different step is that we replace the single measurement
with sequential measurements so that $A^k$ appeared in~\eqref{eq:E_k} can be recovered:
\begin{enumerate}
  \item Generate a quantum state $\rho$.
  \item Using $\rho$ as input, execute the noisy measurement device $k$ times
        \emph{sequentially} and collect the outcome produced by the final $k$-th measurement device.
  \item Repeat the above two steps $M$ rounds.
  \item Output the average $\eta_m^{(k)}$ as an estimate of $E_m^{(k)}$.
\end{enumerate}
Likewise, the average $\overline{\eta}_m:=\sum_{k=1}^{K+1}c_K(k-1) \eta_m^{(k)}$
approximates $\tr[O\rho]$ within error $2\varepsilon$ with probability larger than $1 - \delta$.
The proof is the same as that of Eq.~\eqref{eq:good-estimate}.

A crucial concept we introduce in the above MEM method is the sequential measurement.
Roughly speaking, it means that we use the output of one measurement device as the input of the other.
In Appendix~\ref{appx:sequential measurements}, we elaborate thoroughly this concept and
show that the classical noise model describing the
sequential measurement repeating $k$ times is indeed characterized by the error matrix $A^k$.
For illustrative purpose, we demonstrate in Fig.~\ref{fig:expectation-k3-mem}
the experimental setup for estimating $E_m^{(3)}$,
where the measurement device is executed three times sequentially.
Indeed, one can think of the rightmost $k-1$ measurements as
implementing the calibration subroutine since they always have
the computational basis states as inputs. In some sense, this is \emph{dynamic} calibration
where we do not statically enumerate all computational bases as input states
but dynamically prepare the input states based on the output information
of the target state from the first measurement device.

\begin{figure}[!htbp]
  \centering
  \includegraphics[width=0.4\textwidth]{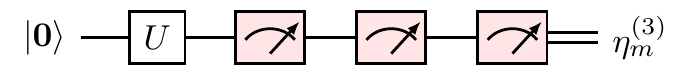}
  \caption{Experimental setup for estimating $E_m^{(3)}$,
        in which the noisy measurement device (box in red) is executed $3$ times sequentially.}
  \label{fig:expectation-k3-mem}
\end{figure}

\textit{\textbf{Resource analysis.}}---When applying the truncated Neumann series framework to mitigate quantum errors,
we repeat the quantum devices sequentially in different numbers of times, compute the noisy values,
and linearly combine them to approximate the target. We use the number of quantum states consumed as the resource metric
and analyze the complexity of the proposed methods GEM and MEM.
Fundamentally, the resource costs of both methods are dominated by the
the optimal truncated number $K$---$K_g$~\eqref{eq:K-opt-gate} in GEM
and $K_m$~\eqref{eq:K-opt-mem} in MEM---that determines the maximal number of truncated terms
and the number of quantum states prepared to evaluate each truncated term.
The detailed analysis has been given in Appendix~\ref{appx:good-estimate}.
Set $\Delta:=\binom{2K+2}{K+1} - 1$.
For each $1\leq k\leq K$, we need $M=2(K+1)\Delta\log(2/\delta)/\varepsilon^2$ copies of quantum states
to achieve the desired accuracy $\varepsilon$ and confidence $\delta$.
As so, the total number of quantum states consumed is roughly given by
\begin{align}
    M(K+1)
&= 2(K+1)^2\Delta\log(2/\delta)/\varepsilon^2 \nonumber \\
&\approx 4^{K}\log(2/\delta)/\varepsilon^2,\label{eq:no-of-states}
\end{align}
where the approximation follows from the Stirling's approximation.
In other words, the number of quantum states consumed by GEM or MEM is much
more than that of the ideal case by a factor of $4^{K}$.
We shall call the factor $4^K$ the error mitigation overhead with truncated Neumann series,
characterizing the overall increased number of samples necessary to ensure a certain accuracy.
At first glance, the exponential factor $4^{K}$ renders the error mitigation methods utilizing
truncated Neumann series infeasible when $K$ becomes large. However, we argue in the following
that for near-term quantum devices with moderate noise resistance, $K$ is quite small
and thus is acceptable experimentally. In Fig.~\ref{fig:K_vs_xi},
we plot the optimal truncated number $K=\lceil\log\varepsilon/\log\xi-1\rceil$
as a function of the noise resistance $\xi$, where the error tolerance parameter is fixed as $\varepsilon=0.01$.
Notice that $\xi$ can be either $\xi_g$~\eqref{eq:noise resistance gem} in GEM (ignoring negligible extra terms)
or $\xi_m$~\eqref{eq:noise resistance mem} in MEM.
One can check from the figure that $K\leq 10$ whenever the noise resistance satisfies $\xi\leq0.657$.
In GEM with Pauli noise, this corresponds to that the minimal Pauli eigenvalue is larger than $0.343$.
In MEM, this corresponds to that the minimal diagonal element of $A$ is larger than $0.67$.
These are mild requirements that many public quantum hardwares fulfill~\cite{chen2019detector,kandala2019error}.
Generally speaking, the incurred error mitigation overhead $4^K$ can be independent of the
system size (the number of qubits), so long as the noise resistance $\xi$ of the quantum device is moderate,
in the sense that it is below a certain threshold (say $0.657$).

\begin{figure}
  \centering
  \includegraphics[width=0.48\textwidth]{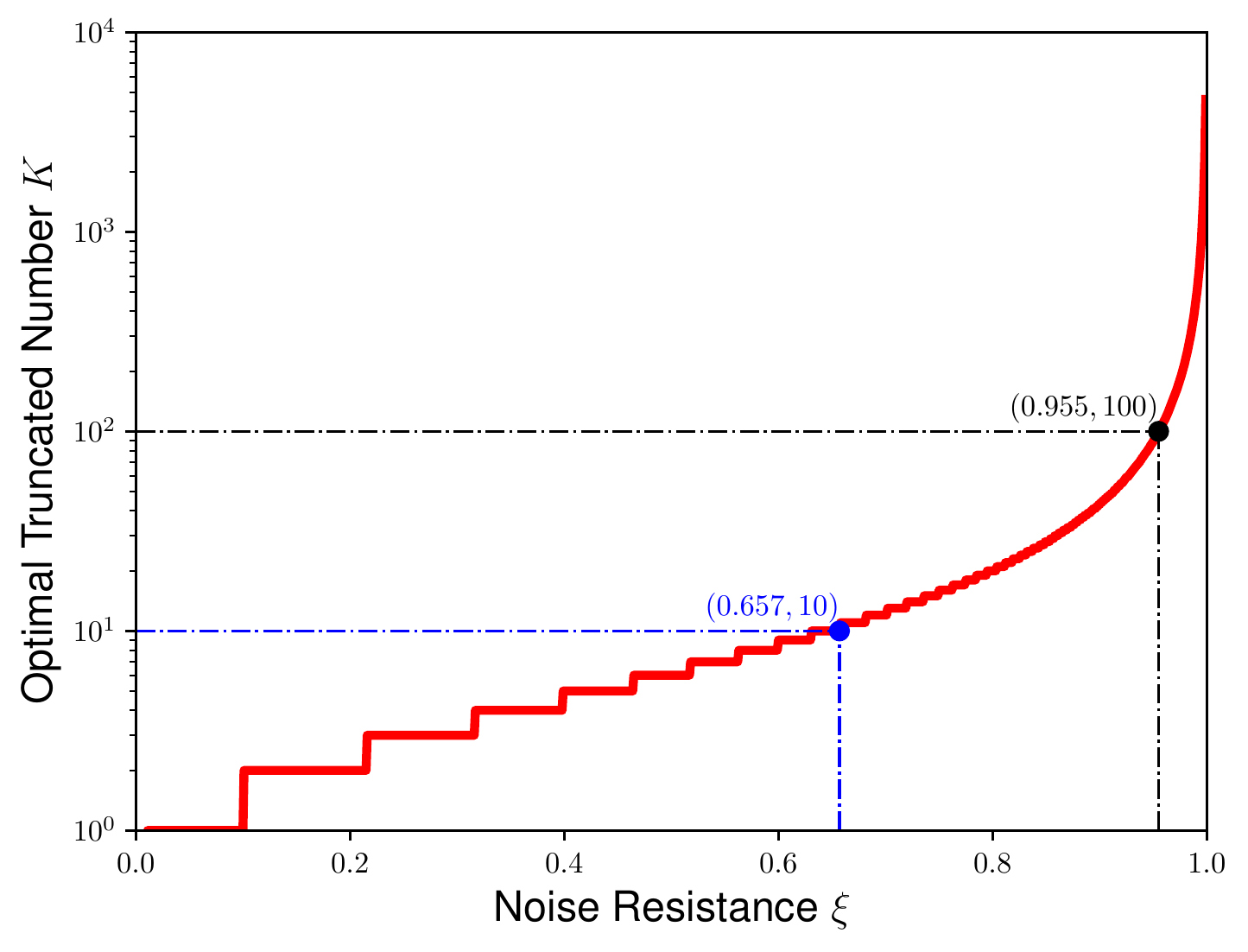}
  \caption{The (simplified) truncated number $K=\lceil\log\varepsilon/\log\xi-1\rceil$
          as a function of the noise resistance $\xi$, where $\varepsilon=0.01$.}
  \label{fig:K_vs_xi}
\end{figure}

\textit{\textbf{Numerical results.}}---We use the qubit depolarizing channel
as example to verify the GEM method, and more examples can be found in Appendix~\ref{appx:examples}.
This channel is defined as $\Omega_p(\rho) := (1-p)\rho + pI/2$, where $p\in[0,1]$.
Consider the task where the ideal state is $\rho = \proj{0}$ and $O=Z$.
The ideal expectation value is $\tr[Z\rho] = 1$ while
the noisy expectation value suffering from depolarizing noise is
$\tr[Z\Omega_p(\rho)] = 1 - p$. To apply GEM, we first compute the optimal truncated number.
Since $\xi_g(\Omega_p) = p$ and $\norm{\llangle Z\vert}{\infty} = 1$,
we get $K_g = \lceil \log\varepsilon/\log p - 1\rceil$.
The mitigated expectation values for different noise level $p$ are
shown in Fig.~\ref{fig:depolarizing}, where $\varepsilon = 0.01$ is fixed.
From the numerical result, we can see that GEM works well and substantially improve the computation accuracy.

\begin{figure}[!htbp]
  \centering
  \includegraphics[width=0.45\textwidth]{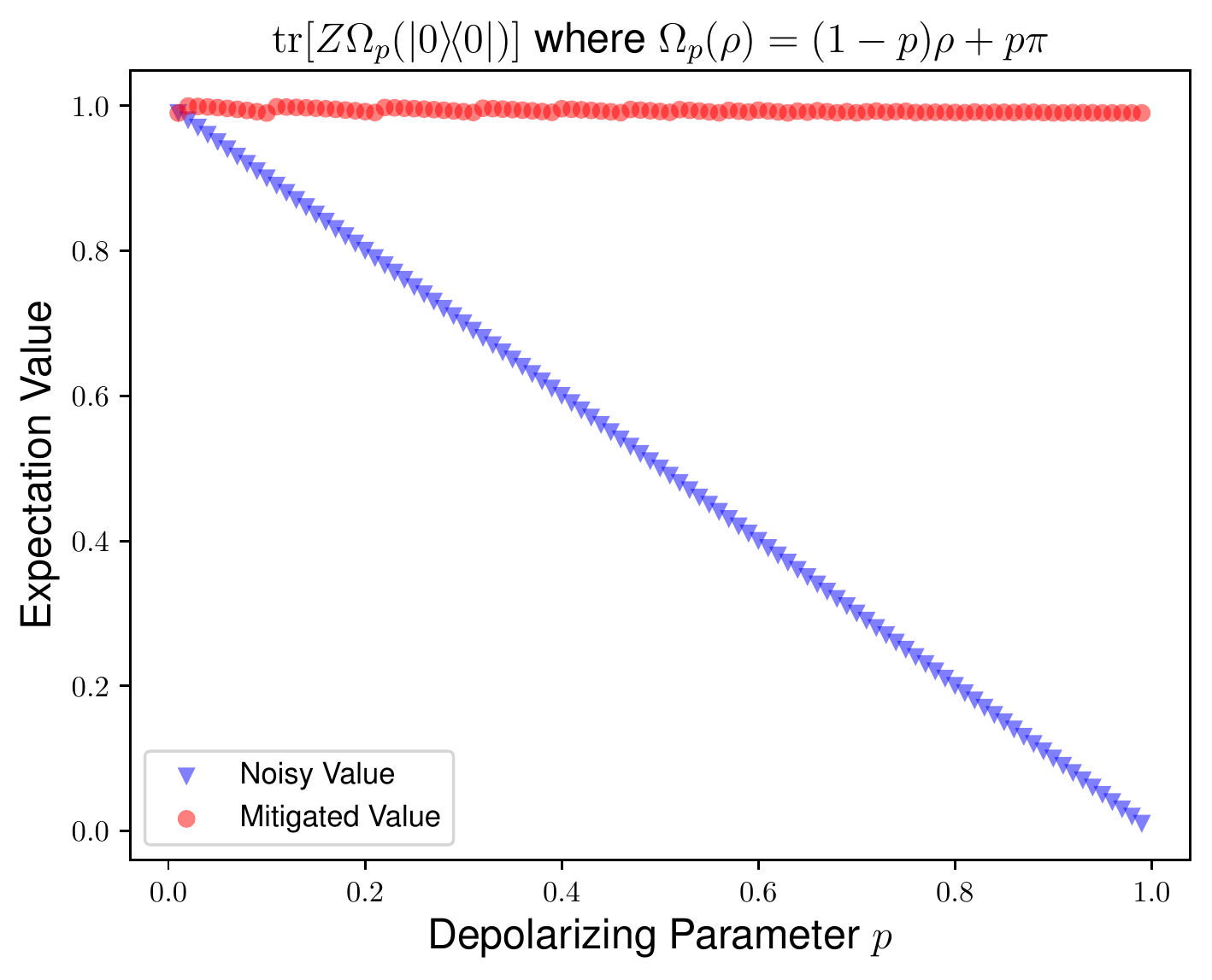}
  \caption{Quantum gate error mitigation via truncated Neumann series for the qubit depolarizing channel,
          where $\varepsilon=0.01$.}
  \label{fig:depolarizing}
\end{figure}

We consider another example to verify the MEM method.
Consider the input state $\rho = \proj{\Phi}$, where $\Phi$ is the maximal superposition state
$\ket{\Phi} := \sum_{i=0}^{2^n-1}\ket{i}/\sqrt{2^n}$.
The observable $O$ is a tensor product of Pauli $Z$ operators, i.e., $O=Z^{\otimes n}$.
The ideal expectation value is $\tr[O\rho] = 0$. We choose $n=8$ and
randomly generate an error matrix $A$ whose noise resistance satisfies
$\xi(A) \approx 0.657$ so that the measurement error fall into the moderate regime.
We repeatedly estimate the noisy expectation values $\eta^{(1)}$ and
compute the the average in a total number of $1000$ times.
Note that all these experiments assume the same error matrix $A$,
and the parameters are fixed as $\varepsilon=\delta=0.01$.
The obtained expectation values are scatted in Fig.~\ref{fig:error-mitigation-result}.
It is easy to see that the noisy measurement device, characterized by the error matrix $A$,
incurs a bias $-0.007$ to the estimated expectation values.
On the other hand, the mitigated expectation values using the MEM method
distributed evenly around the ideal value $0$ within a distance of $0.01$ with high probability.
As evident from Fig.~\ref{fig:error-mitigation-result}, few mitigated values fall outside the expected region.
These statistical outcomes match our conclusions,
validating the correctness and performance of the truncated Neumann series based MEM.

\begin{figure}[!htbp]
  \centering
  \includegraphics[width=0.48\textwidth]{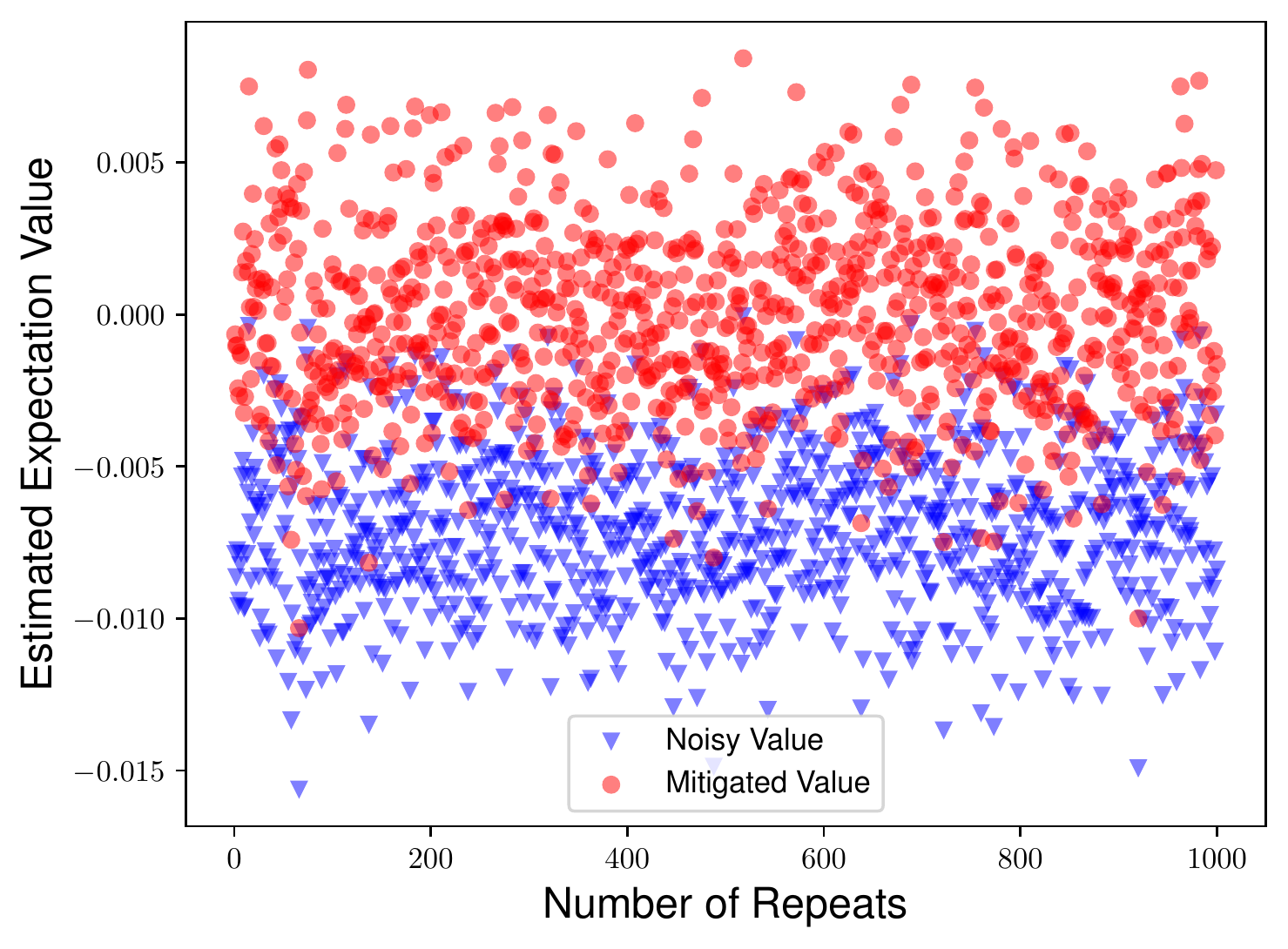}
  \caption{$1000$ noisy estimates $\eta^{(1)}$ (blue triangles)
        and mitigated estimates $\overline{\eta}$ (red dots) via truncated Neumann series
        for the ideal expectation value $\tr[O\rho] = 0$.
        Here, the number of qubits is $8$.}
  \label{fig:error-mitigation-result}
\end{figure}

\textit{\textbf{Conclusions.}}---We introduced a general framework to mitigate quantum gate and measurement errors in
computing expectation values of quantum observables, an essential building block of numerous quantum algorithms.
The idea behind this method is to approximate the inverse of the quantum error characterizing
the noisy behavior of the underlying quantum device using a small number truncated Neumann series terms.
Remarkably, the estimation error decays exponentially in the truncated order, and the incurred error mitigation
overhead is independent of the system size, as long as the noise resistance of the quantum device is moderate.
The proposed error mitigation framework theoretically works for any quantum error and
does not require the tomography procedure to completely characterize the quantum errors.
This property is beneficial and will be more and more important as the quantum circuit
sizes increase. We numerically tested this method for both gate and measurement errors
and found that the computation accuracy is substantially improved.
We believe that this framework will be helpful for quantum error mitigation in NISQ quantum devices.
We emphasize that quantum error mitigation is still an active research area.
For example, there is emerging interest in syncretizing error mitigation and
correction techniques~\cite{piveteau2021error,lostaglio2021error}.
It would be interesting to explore how the proposed error mitigation
framework can be enhanced via error correction.

\textit{\textbf{Acknowledgements.}}---We thank Runyao Duan for his helpful suggestions.




%


\setcounter{secnumdepth}{2}
\appendix
\widetext

\section{Proof of Eq.~\eqref{eq:noisy-expectation}}\label{appx:noisy-expectation}

\begin{proof}
By the definition of $\eta^{(1)}$, we have
\begin{align}
  \eta^{(1)}
&= \frac{1}{M}\sum_{m=1}^MO(\bm{s}^m) \\
&= \frac{1}{M}\sum_{m=1}^M\sum_{\bm{x}\in\{0,1\}^n}O(\bm{x})\langle\bm{x}\vert\bm{s}^m\rangle \\
&= \sum_{\bm{x}\in\{0,1\}^n}O(\bm{x})\langle \bm{x}\vert\left(\frac{1}{M}\sum_{m=1}^M\vert \bm{s}^m\rangle\right).
\end{align}
The expectation value can be evaluated as
\begin{align}
E^{(1)} := \bE[\eta^{(1)}]
&= \bE\left[\sum_{\bm{x}\in\{0,1\}^n}O(\bm{x})\langle \bm{x}\vert
  \left(\frac{1}{M}\sum_{m=1}^M\vert \bm{s}^m\rangle\right)\right] \\
&= \sum_{\bm{x}\in\{0,1\}^n}O(\bm{x})\langle \bm{x}\vert\bE\left[\frac{1}{M}\sum_{m=1}^M\vert \bm{s}^m\rangle\right] \\
&= \sum_{\bm{x}\in\{0,1\}^n}O(\bm{x})\langle \bm{x}\vert A\opn{vec}(\cN(\rho)).
\end{align}
\end{proof}

\section{Proof of Proposition~\ref{prop:Neumann-expansion}}\label{appx:Neumann-expansion}

\begin{proof}
We first recall the definition of Neumann series~\cite[Theorem 4.20]{stewart1998matrix}.
Let $A\in\mathbf{M}_d$ be a real matrix. If $\lim_{k\to\infty}(I-A)^k=0$, then
\begin{align}\label{eq:ShFAGgUTHbU}
    A^{-1} = \sum_{k=0}^\infty (I-A)^k.
\end{align}
Furthermore, a sufficient condition for $\lim_{k\to\infty}(I-A)^k=0$ is that $\spectral{I-A}<1$ in some consistent norm.

Multiplying both sides of Eq.~\eqref{eq:ShFAGgUTHbU} with $A$ and rearranging the elements, we have
\begin{align}\label{eq:ShFAGgUTHbU-2}
I = AA^{-1} = A\left(\sum_{k=0}^K (I-A)^k\right) + A\left(\sum_{k=K+1}^\infty (I-A)^k\right).
\end{align}
The above yields
\begin{align}
  A\left(\sum_{k=K+1}^\infty (I-A)^k\right)
&= I - A\left(\sum_{k=0}^K (I-A)^k\right) \\
&= I - \frac{(I-A)^0 \times [I - (I-A)^{K+1}]}{I - (I - A)}A \\
&= (I-A)^{K+1},
\end{align}
where the second equality follows from the closed-form formula of geometric series.
On the other hand, using the recurrence relation for binomial coefficients we can show that
\begin{align}\label{eq:ShFAGgUTHbU-3}
  \sum_{k=0}^K (I-A)^k = \sum_{k=0}^K \binom{K+1}{k+1}(-A)^k = \sum_{k=0}^K c_K(k) A^k,
\end{align}
where the coefficient function $c_K$ is defined in Eq.~\eqref{eq:c_k}.
Putting all pieces together, we get
\begin{align}
    I = \sum_{k=1}^{K+1} c_K(k-1) A^k + (I-A)^{K+1},
\end{align}
where we change the variable so that the summation begins with $k=1$.
Thus, using the linearity of function $f$ we conclude that
\begin{align}
    f(I) = \sum_{k=1}^{K+1} c_K(k-1) f\left(A^k\right) + f\left((I-A)^{K+1}\right).
\end{align}
\end{proof}

\section{Pauli transfer matrix representation} 
\label{appx:pauli_transfer_matrix_representation}

The four (normalized) Pauli operators in the qubit space are defined as
\begin{subequations}\label{eq:Pauli-qubit}
\begin{alignat}{3}
  I \equiv \sigma_0 &:= \frac{1}{\sqrt{2}}\begin{pmatrix} 1 & 0 \\ 0 & 1\end{pmatrix},
   && \quad &
  X \equiv \sigma_1 &:= \frac{1}{\sqrt{2}}\begin{pmatrix} 0 & 1 \\ 1 & 0\end{pmatrix}, \\
  Y \equiv \sigma_2 &:= \frac{1}{\sqrt{2}}\begin{pmatrix} 0 & -i \\ i & 0\end{pmatrix},
   && \quad &
  Z \equiv \sigma_3 &:= \frac{1}{\sqrt{2}}\begin{pmatrix} 1 & 0 \\ 0 & -1\end{pmatrix}.
\end{alignat}
\end{subequations}
They provide an orthonormal basis for the qubit linear operators, i.e.,
arbitrary qubit linear operator can be decomposed with respect to this basis.
For the $n$-qubit case, one can construct a set of Pauli operators, which we call the Pauli set, as
\begin{align}\label{eq:Pauli-set}
\bm{P}_n := \left\{\bigotimes_{k=1}^n\sigma_{j_k} : j_k = 0,1,2,3\right\}
            \equiv \left\{I, X, Y, Z\right\}^{\ox n}.
\end{align}
It is easy to verify that $\vert\bm{P}_n\vert=4^n$, where $\vert\cdot\vert$ denotes the size of the set.

The Pauli transfer matrix (PTM) representation expresses states and evolutions in terms of the
Pauli basis $\bm{P}_n$~\eqref{eq:Pauli-set}, since $\bm{P}_n$ forms a basis for the $n$-qubit operators.
Specially, an $n$-qubit quantum state $\rho$ can be written in the vector form by
decomposing it into the Pauli basis:
\begin{align}
  \vert\rho\rrangle := \begin{bmatrix} \vdots \\ \rho_{\bm{i}} \\ \vdots \end{bmatrix},
\end{align}
where the vector element is $\rho_{\bm{i}} := \tr[P_{\bm{i}}\rho]\in\mathbb{R}$ and $P_{\bm{i}}\in\bm{P}_n$.
Since $\vert\bm{P}_n\vert = 4^n$, $\vert\rho\rrangle$ is a $4^n$-dimensional real column vector.
In this representation, a basis for the vector space is
\begin{align}
\left\{\vert P_{\bm{i}} \rrangle: \bm{i} \in \{0,1,2,3\}^n \right\}.
\end{align}
Similarly, a Hermitian observable $O$ can also be expressed as a $4^n$-dimensional real row vector via
\begin{align}
  \llangle O \vert := [\cdots~~ O_{\bm{i}}~~ \cdots],
\end{align}
where the vector element is $O_{\bm{i}}:=\tr[OP_{\bm{i}}]$ and $P_{\bm{i}}\in\bm{P}_n$.
That is to say, $O$ can be written as a linear combination of Pauli operators $O=\sum_{\bm{i}}O_{\bm{i}}P_{\bm{i}}$.
By definition, one has $\norm{\llangle O\vert}{\infty} = \sum_{\bm{i}}\vert O_{\bm{i}}\vert$,
where $\norm{\cdot}{\infty}$ is the matrix $\infty$-norm.
In quantum computation, a common choice of the observable is $O=Z^{\ox n}$.
In this case, one has $\norm{\llangle Z^{\ox n}\vert}{\infty}=1$.

A quantum channel $\cN$ (with $n$-qubit input and output systems) can be expressed as a
$4^n\times 4^n$ real square matrix $[\cN]$ whose
$\bm{i}$-th row and $\bm{j}$-th column element is defined via
\begin{align}\label{eq:PTM-channel}
  [\cN]_{\bm{i},\bm{j}} := \tr\left[P_{\bm{i}}\cN(P_{\bm{j}})\right].
\end{align}
Here, $\bm{i},\bm{j}\in\{0,1,2,3\}^{n}$.
Note that the identity channel $\id$ has identity Pauli transfer matrix, i.e., $[\id] = \1$,
as evident from the relation $\tr[P_{\bm{i}}P_{\bm{j}}] = \delta_{\bm{i},\bm{j}}$, where $\delta_{x,y}$ is the
Kronecker delta function.
By definition, if $\sigma = \cN(\rho)$, then $\vert\sigma\rrangle = [\cN]\vert\rho\rrangle$.

A desirable property of the Pauli transfer matrix representation (PTM) is that
quantum channel composition is multiplicative in this representation. Specifically,
let $\cM$ and $\cN$ be two $n$-qubit quantum channels, then
\begin{align}\label{eq:multiplicative}
    \left[\cM\circ\cN\right] = \left[\cM\right]\left[\cN\right].
\end{align}
The proof is given as follow.

\begin{proof}
Let $\bm{i},\bm{j}\in\{0,1,2,3\}^n$. By definition,
\begin{align}
  \left[\cM\circ\cN\right]_{\bm{i},\bm{j}}
&= \tr\left[P_{\bm{i}} \cM\circ\cN(P_{\bm{j}})\right] \\
&= \tr\left[P_{\bm{i}} \cM\left(\sum_{\bm{m}}[\cN]_{\bm{m},\bm{j}}P_{\bm{m}}\right)\right] \\
&= \tr\left[P_{\bm{i}} \sum_{\bm{m},\bm{n}}[\cN]_{\bm{m},\bm{j}}[\cM]_{\bm{n},\bm{m}} P_{\bm{n}}\right] \\
&= \tr\left[P_{\bm{i}} \sum_{n}\left(\sum_{\bm{m}}[\cM]_{\bm{n},\bm{m}}
                                            [\cN]_{\bm{m},\bm{j}}\right)P_{\bm{n}}\right] \\
&= \sum_{\bm{n}}\left([\cM][\cN]\right)_{\bm{n},\bm{j}}\tr\left[P_{\bm{i}}P_{\bm{n}}\right] \\
&= \left([\cM][\cN]\right)_{\bm{i},\bm{j}},
\end{align}
where in the last step we use the fact that $\tr[P_{\bm{i}}P_{\bm{j}}] = \delta_{\bm{i},\bm{j}}$.
\end{proof}

\begin{lemma}\label{eq:rho-bound}
Let $\rho$ be a $n$-qubit quantum state.
For arbitrary $\bm{i} \in\{0,1,2,3\}^n$, it holds that
$\vert \llangle P_{\bm{i}} \vert \rho \rrangle \vert \leq 1$.
\end{lemma}
\begin{proof}
By definition, $\llangle P_{\bm{i}} \vert \rho \rrangle = \tr\left[P_{\bm{i}}\rho\right]$.
Notice that each Pauli operator $P_{\bm{i}}$ has eigenvalues $\pm 1$,
this implies that $\tr\left[P_{\bm{i}}\rho\right]$ is real.
Assume the eigenstates of $P_{\bm{i}}$ is $\{\ket{v_j}\}_j$.
We expand $\rho$ in this basis: $\rho = \sum_{jk}\rho_{kj}\ket{v_j}\!\bra{v_k}$.
Then
\begin{align}
    \vert\tr\left[P_{\bm{i}}\rho\right]\vert
=   \left\vert\sum_{jk}\rho_{kj}\bra{v_k}P_{\bm{i}}\ket{v_j}\right\vert
\leq \sum_{j}\rho_{jj}
= \tr[\rho] = 1,
\end{align}
where the inequality follows from the fact that $P_{\bm{i}}$ has eigenvalues $\pm 1$.
\end{proof}


\section{Proof of Theorem~\ref{thm:GEM}}\label{appx:GEM}

\begin{proof}[Proof of Theorem~\ref{thm:GEM}]
Using Proposition~\ref{prop:Neumann-expansion}, we have
\begin{align}
  \left\vert \tr[O\rho] - \sum_{k=1}^{K+1}c_K(k-1) E^{(k)}\right\vert
&=   \left\vert \llangle O\vert (I - [\cN])^{K+1} \vert\rho\rrangle\right\vert.
\end{align}
We need to upper bound the tail $\left\vert \llangle O\vert (I - [\cN])^{K+1} \vert\rho\rrangle\right\vert$.
Consider the following chain of inequalities:
\begin{align}
 \left\vert \llangle O\vert (I - [\cN])^{K+1} \vert\rho\rrangle\right\vert
&=   \left\vert\llangle O \left(\sum_{\bm{i}} \vert P_{\bm{i}}\rrangle\!\llangle P_{\bm{i}}\vert\right)
      (I - [\cN])^{K+1}
      \left(\sum_{\bm{j}} \vert P_{\bm{j}}\rrangle\!\llangle P_{\bm{j}}\vert\right)
      \vert\rho\rrangle \right\vert \\
&=   \left\vert \sum_{\bm{i},\bm{j}} \llangle O\vert P_{\bm{i}}\rrangle \llangle P_{\bm{j}}\vert\rho\rrangle
    \llangle P_{\bm{i}} \vert (I - [\cN])^{K+1} \vert P_{\bm{j}} \rrangle \right\vert \\
&\leq \sum_{\bm{i},\bm{j}}
      \vert \llangle O\vert P_{\bm{i}} \rrangle \vert
      \cdot
      \vert \llangle P_{\bm{j}}\vert\rho\rrangle \vert
      \cdot
      \vert \llangle P_{\bm{i}} \vert (I - [\cN])^{K+1} \vert P_{\bm{j}} \rrangle \vert \\
&\leq \sum_{\bm{i},\bm{j}}
      \vert \llangle O\vert P_{\bm{i}} \rrangle \vert
      \cdot
      \vert \llangle P_{\bm{i}} \vert (I - [\cN])^{K+1} \vert P_{\bm{j}} \rrangle \vert\label{eq:WIkgPiKRP1} \\
&= \sum_{\bm{i}}\vert \llangle O\vert P_{\bm{i}} \rrangle \vert
      \left(\sum_{\bm{j}} \vert \llangle P_{\bm{i}} \vert (I - [\cN])^{K+1} \vert P_{\bm{j}}\rrangle\vert\right) \\
&\leq \sum_{\bm{i}}\vert \llangle O\vert P_{\bm{i}} \rrangle \vert\norm{(I - [\cN])^{K+1}}{\infty}\label{eq:WIkgPiKRP2} \\
&= \norm{\llangle O\vert}{\infty}\norm{(I - [\cN])^{K+1}}{\infty} \\
&\leq \norm{\llangle O\vert}{\infty}\norm{I - [\cN]}{\infty}^{K+1},\label{eq:WIkgPiKRP3}
\end{align}
where Eq.~\eqref{eq:WIkgPiKRP1} follows from Lemma~\ref{eq:rho-bound},
Eq.~\eqref{eq:WIkgPiKRP2} follows from the definition of the matrix $\infty$-norm,
and Eq.~\eqref{eq:WIkgPiKRP3} follows from the submultiplicativity of the matrix $\infty$-norm.
\end{proof}

\section{Proof of Eq.~\eqref{eq:good-estimate}}\label{appx:good-estimate}

\begin{proof}
Let $\bm{s}^{m,k}\in\{0,1\}^n$ be the outcome in the $m$-th round when estimating $E_g^{(k)}$. By definition,
\begin{align}
\overline{\eta}
&= \sum_{k=1}^{K+1} c_K(k-1) \eta^{(k)} \\
&= \frac{1}{M}\sum_{k=1}^{K+1}\sum_{m=1}^{M}c_K(k-1) O(\bm{s}^{m,k}) \\
&= \frac{1}{M(K+1)}\sum_{k=1}^{K+1}\sum_{m=1}^{M}(K+1)c_K(k-1) O(\bm{s}^{m,k}).
\end{align}
Introducing the new random variables $X_{m,k}:= (K+1)c_K(k-1)O(\bm{s}^{m,k})$, we have
\begin{align}
\overline{\eta} = \frac{1}{M(K+1)}\sum_{k=1}^{K+1}\sum_{m=1}^{M} X_{m,k}.\label{eq:TgJvNAKG}
\end{align}
Intuitively, Eq.~\eqref{eq:TgJvNAKG} says that $\eta$ can be viewed as the empirical mean value
of the set of random variables
\begin{align}
    \left\{X_{m,k}: m=1,\cdots,M; k=1,\cdots,K+1\right\}.
\end{align}
First, we show that the absolute value of each $X_{m,k}$ is upper bounded as
\begin{align}\label{eq:mMZgBWZvH1}
  \vert X_{m,k}\vert = \vert (K+1)c_K(k-1)O(\bm{s}^{m,k}) \vert
  \leq (K+1) \vert c_K(k-1)\vert \vert O(\bm{s}^{m,k}) \vert
  \leq (K+1) \vert c_K(k-1)\vert,
\end{align}
where the second inequality follows from the assumption of $O$
(recall that $O$ is diagonal in the computational basis and $\norm{O}{2}\leq 1$).
Then, we show that $\overline{\eta}$ is an unbiased estimator of the quantity $\sum_{k=1}^{K+1}c_K(k-1)E^{(k)}$:
\begin{subequations}\label{eq:mMZgBWZvH2}
\begin{align}
  \bE[\overline{\eta}]
&= \bE\left[\frac{1}{M(K+1)}\sum_{k=1}^{K+1}\sum_{m=1}^{M} X_{m,k}\right] \\
&= \bE\left[\frac{1}{M}
   \sum_{k=1}^{K+1}\sum_{m=1}^{M}c_K(k-1)O(\bm{s}^{m,k})\right] \\
&= \sum_{k=1}^{K+1} c_K(k-1)\left(\sum_{\bm{x}}O(\bm{x})\langle \bm{x}\vert
      \bE_M\left[ \frac{1}{M} \sum_{m=1}^M \vert \bm{s}^{m,k}\rangle \right]\right) \\
&= \sum_{k=1}^{K+1} c_K(k-1)\left(
    \sum_{\bm{x}}O(\bm{x})\langle \bm{x}\vert A^k\opn{vec}(\rho)\right) \\
&= \sum_{k=1}^{K+1}c_K(k-1) E^{(k)},
\end{align}
\end{subequations}
where the last equality follows from~\eqref{eq:E_k}.
Eqs.~\eqref{eq:mMZgBWZvH1} and~\eqref{eq:mMZgBWZvH2} together guarantee that
the prerequisites of the Hoeffding's inequality hold. By the Hoeffding's equality, we have
\begin{align}
    \Pr\left\{\left\vert\overline{\eta} - \sum_{k=1}^{K+1}c_K(k-1) E^{(k)} \right\vert \geq \varepsilon\right\}
&\leq 2\exp\left(- \frac{2M^2(K+1)^2\varepsilon^2}{ 4\sum_{k=1}^{K+1}\sum_{m=1}^M((K+1)c_K(k))^2}\right) \\
&= 2\exp\left(- \frac{2M^2(K+1)^2\varepsilon^2}{ 4M(K+1)^3\left(\sum_{k=0}^K [c_K(k)]^2\right)}\right) \\
&= 2\exp\left(- \frac{M\varepsilon^2}{ 2(K+1)\Delta}\right),
\end{align}
where $\Delta:=\sum_{k=1}^{K+1} [c_K(k)]^2 = \binom{2K+2}{K+1} - 1$.
Solving
\begin{align}
    2\exp\left(- \frac{M\varepsilon^2}{ 2(K+1)\Delta}\right) \leq \delta
\end{align}
gives
\begin{align}
    M \geq 2(K+1)\Delta\log(2/\delta)/\varepsilon^2.
\end{align}
To summarize, choosing $K =\left\lceil \log\varepsilon/\log\xi - 1\right\rceil$
and $M =\lceil 2(K+1)\Delta\log(2/\delta)/\varepsilon^2 \rceil$, we are able obtain the following
two statements
\begin{align}
    \Pr\left\{\left\vert\overline{\eta} - \sum_{k=1}^{K+1}c_K(k-1) E^{(k)} \right\vert \geq \varepsilon\right\} \leq \delta, \\
    \left\vert \tr[O\rho] - \sum_{k=1}^{K+1}c_K(k-1) E^{(k)}\right\vert \leq \varepsilon,
\end{align}
where the first one is shown above and the second one is proved in Theorem~\ref{thm:GEM}.
Using the union bound and the triangle inequality, we conclude that $\overline{\eta}$
can estimate the ideal expectation value $\tr[O\rho]$ with error $2\varepsilon$ at a probability
greater than $1-\delta$.
\end{proof}

\section{Proof of Theorem~\ref{thm:MEM}}\label{appx:MEM}

\begin{proof}
Using Proposition~\ref{prop:Neumann-expansion}, we have
\begin{align}
  \left\vert \tr[O\rho] - \sum_{k=1}^{K+1}c_K(k-1) E^{(k)}\right\vert
= \left\vert \sum_{\bm{x}\in\{0,1\}^n}O(\bm{x})\langle \bm{x}\vert(I-A)^{K+1}\opn{vec}(\rho)\right\vert.
    \label{eq:appx:MEM-2}
\end{align}
Now we show that the quantity in~\eqref{eq:appx:MEM-2} can be bounded from above.
Define the matrix $1$-norm of a $m\times n$ matrix $B$ as
\begin{align}
   \norm{B}{1} := \max_{1\leq j \leq n}\sum_{i=1}^n\vert B_{ij}\vert
                \equiv \max_{1\leq j \leq n} \sum_{i=1}^n \vert \bra{i} B \ket{j}\vert,
\end{align}
which is simply the maximum absolute column sum of the matrix.
Let $\rho(\bm{y})$ is the $\bm{y}$-th diagonal element of the quantum state $\rho$.
Consider the following chain of inequalities:
\begin{subequations}\label{eq:bounding-tail}
\begin{align}
\left\vert \sum_{\bm{x}\in\{0,1\}^n}O(\bm{x})\langle \bm{x}\vert(I-A)^{K+1}\opn{vec}(\rho)\right\vert
&= \left\vert\sum_{\bm{x}\in\{0,1\}^n}\sum_{\bm{y}\in\{0,1\}^n}O(\bm{x})
    \rho(\bm{y})\langle \bm{x}\vert (I-A)^{K+1}\vert\bm{y}\rangle\right\vert\\
&\leq\sum_{\bm{x}\in\{0,1\}^n}\sum_{\bm{y}\in\{0,1\}^n}
    \vert O(\bm{x})\vert \cdot
    \rho(\bm{y}) \cdot \left\vert\langle \bm{x}\vert(I-A)^{K+1}\vert\bm{y}\rangle\right\vert\\
&\leq\sum_{\bm{x}\in\{0,1\}^n}\sum_{\bm{y}\in\{0,1\}^n}
    \rho(\bm{y})\left\vert\langle \bm{x}\vert(I-A)^{K+1}\vert\bm{y}\rangle\right\vert\label{eq:appx:MEM-3}\\
&=\sum_{\bm{y}\in\{0,1\}^n}\rho(\bm{y})
  \sum_{\bm{x}\in\{0,1\}^n}\left\vert\langle \bm{x}\vert(I-A)^{K+1}\vert\bm{y}\rangle\right\vert \\
&\leq\sum_{\bm{y}\in\{0,1\}^n} \rho(\bm{y}) \lVert (I-A)^{K+1} \lVert_1\label{eq:appx:MEM-4} \\
&= \lVert (I-A)^{K+1} \lVert_1 \label{eq:appx:MEM-5} \\
&\leq \lVert I-A \lVert_1^{K+1}\label{eq:appx:MEM-6} \\
&= \xi_m^{K+1},\label{eq:appx:MEM-7}
\end{align}
\end{subequations}
where~\eqref{eq:appx:MEM-3} follows from the assumption
that $O$ is diagonalized in the computational basis and $\norm{O}{2} \leq 1$,
~\eqref{eq:appx:MEM-4} follows from the definition of matrix $1$-norm,
~\eqref{eq:appx:MEM-5} follows from the fact that $\rho$ is a quantum state
and thus $\sum_{\bm{y}}\rho(\bm{y})=1$,
~\eqref{eq:appx:MEM-6} follows from the submultiplicativity property of the matrix $1$-norm,
and~\eqref{eq:appx:MEM-7} follows from Lemma~\ref{lemma:min-norm} stated below.
We are done.
\end{proof}

\begin{lemma}\label{lemma:min-norm}
Let $A$ be a column stochastic matrix of size $d \times d$. It holds that
\begin{align}
    \xi_m(A) = \norm{I - A}{1},
\end{align}
where $\xi_m(A)$ is defined in~\eqref{eq:noise resistance mem}.
\end{lemma}
\begin{proof}
Since $A$ is column stochastic, $I-A$ has non-negative diagonal elements and negative off-negative elements.
Thus
\begin{align}
  \norm{I - A}{1}
&= \max_{1\leq j\leq d}\left(1 - A_{jj} + \sum_{i\neq j} A_{ij}\right) \\
&= \max_{1\leq j\leq d}\left(1 - A_{jj} + 1 - A_{jj}\right) \\
&= 2\max_{1\leq j\leq d}\left(1 - A_{jj}\right) \\
&= 2 - 2\min_{1\leq j\leq d} A_{jj} \\
&=: \xi_m(A),
\end{align}
where the second equality follows from the fact that $A$ is column stochastic.
\end{proof}

\section{Sequential measurements}\label{appx:sequential measurements}

In this Appendix, we prove that the classical noise model describing the
sequential measurement repeating $k$ times is effectively characterized by the
stochastic matrix $A^k$. We begin with the simple case $k=2$.
Since the noise model is classical and linear in the input,
it suffices to consider the computational basis states as inputs.
As shown in Fig.~\ref{fig:sequential-measurement},
we apply the noisy quantum measurement device two times sequentially
on the input state $\proj{\bm{x}}$ in computational basis
where $\bm{x}\in\{0,1\}^{n}$.
Assume the measurement outcome of the first measurement is $\bm{y}$ and
the measurement outcome of the second measurement is $\bm{z}$,
where $\bm{y},\bm{z}\in\{0,1\}^{n}$.
Assume that the error matrix associated with this sequential measurement is
$A'$. That is, the probability of obtaining the outcome $\bm{z}$ provided
the true outcome is $\bm{x}$ is given by $A'_{\bm{z}\bm{x}}$.
Practically, we input $\proj{\bm{x}}$ to the first noisy measurement device
and obtain the outcome $\bm{y}$. The probability of this event
is $A_{\bm{y}\bm{x}}$, by the definition of the error matrix.
Similarly, we input $\proj{\bm{y}}$ to the second noisy measurement device
and obtain the outcome $\bm{z}$. The probability of this event
is $A_{\bm{z}\bm{y}}$. Inspecting the chain $\bm{x}\to\bm{y}\to\bm{z}$, we have
\begin{align}
    A'_{\bm{z}\bm{x}}
= \sum_{\bm{y}\in\{0,1\}^{n}} A_{\bm{y}\bm{x}}A_{\bm{z}\bm{y}}
= A^2_{\bm{z}\bm{x}}.
\end{align}
The above analysis justifies that the classical noise model describing the sequential
measurement repeating $2$ times is effectively characterized by the
stochastic matrix $A^2$. The general case can be analyzed similarly.

\begin{figure}[!htbp]
  \centering
  \includegraphics[width=0.3\textwidth]{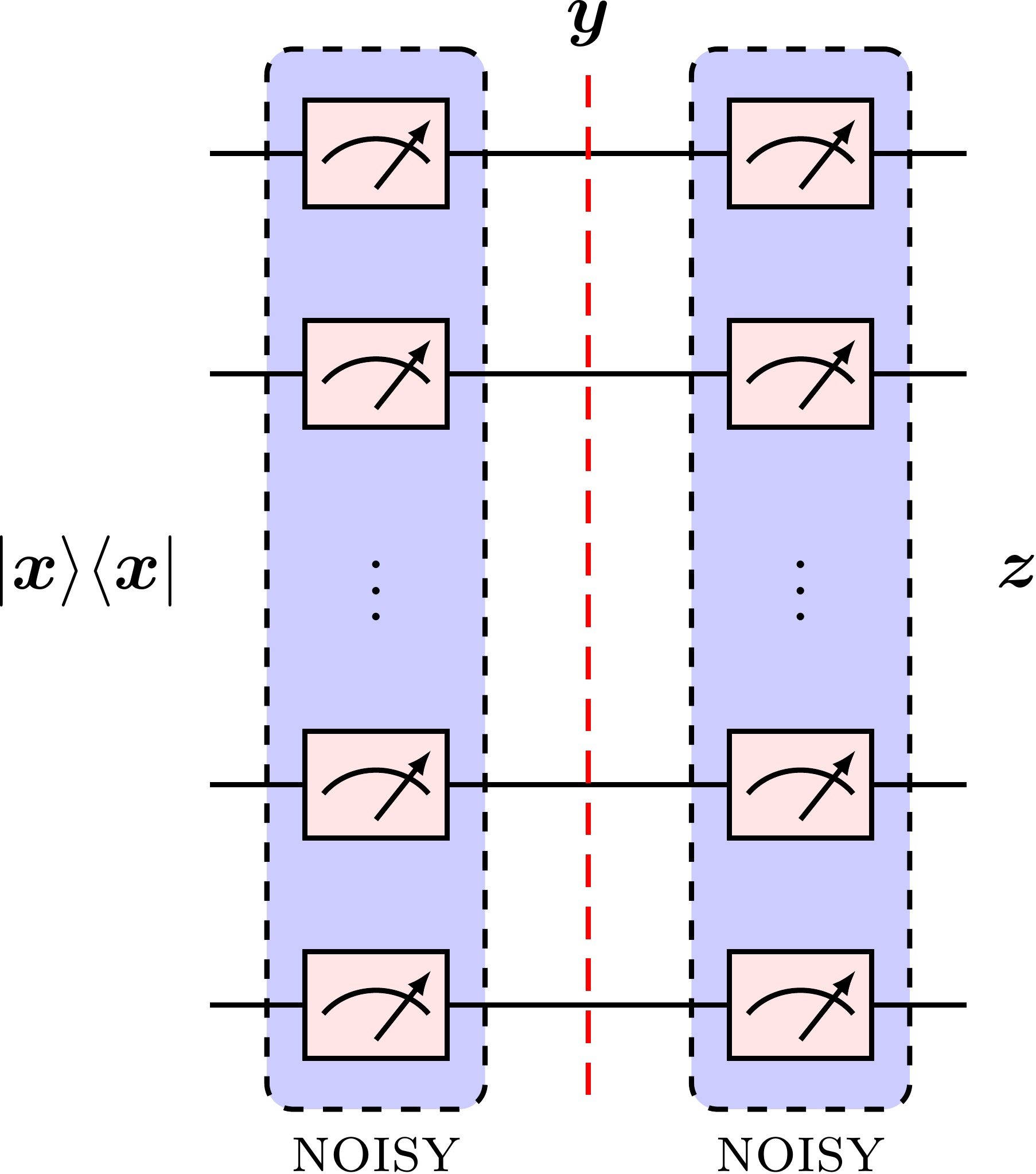}
  \caption{Apply the noisy quantum measurement device two times sequentially
        on the input state $\proj{\bm{x}}$ where $\bm{x}\in\{0,1\}^{n}$.
        The measurement outcome of the first measurement is $\bm{y}$ and
        the measurement outcome of the second measurement is $\bm{z}$.}
  \label{fig:sequential-measurement}
\end{figure}

Mathematically, quantum measurements can be modeled as
quantum-classical quantum channels~\cite[Chapter 4.6.6]{wilde2016quantum}
where they take a quantum system to a classical one.
Experimentally, the implementation of quantum measurement
is platform-dependent and has different characterizations.
For example, the fabrication and control of quantum coherent superconducting circuits
have enabled experiments that implement quantum measurement~\cite{naghiloo2019introduction}.
Based on the outcome data, experimental measurements are typically categorized
into two types: those only output classical outcomes
and those output both classical outcomes and quantum states.
That is, besides the usually classical outcome sequences,
the measurement device will also output a quantum state on the computational basis
corresponding to the classical outcome.
For the former type, we can implement the sequential measurement
via the \emph{qubit reset}~\cite{egger2018pulsed,magnard2018fast,yirka2020qubit} approach,
by which we mean the ability to re-initialize the qubits into a known state,
usually a state in the computational basis, during the course of the computation.
Technically, when the $i$-th noisy measurement outputs an outcome
sequence $\bm{s}^i\in\{0,1\}^n$, we use the qubit reset technique to prepare
the computational basis state $\vert\bm{s}^i\rangle\!\langle\bm{s}^i\vert$ and feed it to
the $(i+1)$-th noisy measurement (cf. Fig.~\ref{fig:sequential-measurement}).
In this case, the noisy measurement device can be reused.
For the latter type, the sequential measurement can be implemented efficiently:
when the $i$-th noisy measurement outputs a classical sequence and a quantum state
on the computational basis, we feed the quantum state to the $(i+1)$-th noisy measurement.

\section{Demonstrative examples}\label{appx:examples}

\textbf{Dephasing noise.}
The qubit dephasing error can arise when the energy splitting of a qubit fluctuates
as a function of time due to coupling to the environment. Charge noise affecting
a transmon qubit is of this type. Dephasing is represented by a phase-flip channel,
which describes the loss of phase information with probability $p$. This
channel projects the state onto the $Z$-axis of the Bloch sphere with probability $p$,
and does nothing with probability $1-p$:
\begin{align}\label{eq:dephasing-channel}
   \cD_p(\rho) := \left(1 - p\right)\rho + pZ \rho Z,
\end{align}
where $p\in[0,1]$. The PTM representation of $\cD_p$ is~\cite{greenbaum2015introduction}
\begin{align}
   [\cD_p] =
\begin{pmatrix}
1 & 0 & 0 & 0 \\
0 & 1-2p & 0 & 0 \\
0 & 0 & 1-2p & 0 \\
0 & 0 & 0 & 1
\end{pmatrix}
\end{align}
and thus $\xi(\cD_p) = 2p$. The noise resistance implies that our method fails
whenever $p\geq1/2$, for which the truncation error cannot be bounded any more.

\textbf{Amplitude damping noise.} The amplitude damping channel is defined as
\begin{align}\label{eq:AD}
  \cA_\gamma(\rho) := E_1 \rho E_1^\dag + E_2 \rho E_2^\dag,
\end{align}
where $\gamma\in[0,1]$, $E_1:=\proj{0} + \sqrt{1-\gamma}\proj{1}$, and $E_2:=\sqrt{\gamma}\ketbra{0}{1}$.
$\gamma$ can be thought of as the probability of losing a photon.
The PTM representation of $\cA_\gamma$ is~\cite{greenbaum2015introduction}
\begin{align}\label{eq:AD-PTM}
   [\cA_\gamma] =
\begin{pmatrix}
1 & 0 & 0 & 0 \\
0 & \sqrt{1-\gamma} & 0 & 0 \\
0 & 0 & \sqrt{1-\gamma} & 0 \\
\gamma & 0 & 0 & 1-\gamma
\end{pmatrix}
\end{align}
and thus $\xi(\cA_\gamma) = 2\gamma$. Likewise,
our method fails for the amplitude damping channel
whose noise parameter $\gamma$ exceeds $1/2$, for which the truncation error cannot be bounded.

\end{document}